\documentclass{article}
\usepackage{graphicx,amssymb,amsmath}

\usepackage{color}
\usepackage{amsfonts}
\usepackage{latexsym}
\usepackage{fullpage}

\usepackage[usenames,dvipsnames]{xcolor}

\usepackage[ruled]{algorithm}
\usepackage[noend]{algpseudocode}

\usepackage{comment}

\newcommand{\ignore}[1]{}

\usepackage{enumitem}

\usepackage[section]{placeins}

\usepackage{subfig}
\usepackage{setspace}

\algdef{SE}[DOWHILE]{Do}{doWhile}{\algorithmicdo}[1]{\algorithmicwhile\ #1}%

\newtheorem{definition}{Definition}

\newtheorem{lemma}{Lemma}
\newtheorem{theorem}{Theorem}
\newtheorem{corollary}{Corollary}

\newtheorem{proposition}{Proposition}

\newenvironment{proof}{\noindent{\bf Proof.}}{\hfill \qed \vskip 5pt}
\def\qed{\hfill\rule{2mm}{2mm}}

\begin{document}

\title{Search on a Line by Byzantine Robots\footnote{This is the full version of the paper with the same title which will appear in the proceedings of the 27th International Symposium on Algorithms and Computation (ISAAC), December 12-14, Sydney, Australia.}
}

\author{J. Czyzowicz\thanks{Research supported in part by NSERC.}
\thanks{D\'{e}partement d'informatique, Universit\'{e} du Qu\'{e}bec en Outaouais, Gatineau, QC, Canada, {\tt jurek.czyzowicz@uqo.ca }}
\and K. Georgiou$^\dag$
\thanks{Department of Mathematics, Ryerson University, Toronto, ON, Canada
{\tt konstantinos@ryerson.ca }
}
\and E. Kranakis$^\dag$
\thanks{School of Computer Science, Carleton University, Ottawa, ON, Canada
{\tt kranakis@scs.carleton.ca}
}
\and D. Krizanc
\thanks{Department of Mathematics and Computer Science, Wesleyan University, Middletown CT, USA
{\tt  dkrizanc@wesleyan.edu}
}
\and L. Narayanan$^\dag$
\thanks{Department of Computer Science and Software Engineering, Concordia University, Montreal, QC,  Canada
{\tt lata@cs.concordia.ca }
}
\and J. Opatrny$^\dag$
\thanks{Department of Computer Science and Software Engineering,
Concordia University, Montreal, QC,  Canada
{\tt opatrny@cs.concordia.ca }
}
\and S. Shende
\thanks{Department of Computer Science, Rutgers University, Camden, USA
{\tt sunil.shende@rutgers.edu }
}}

\thispagestyle{empty}
\maketitle

\begin{abstract}
We consider the problem of fault-tolerant parallel search on an
infinite line by $n$ robots. Starting from the origin, the robots
are required to find a target at an unknown location. The robots
can move with maximum speed $1$ and can communicate in wireless mode
among themselves. However, among the $n$ robots, there are $f$
robots that exhibit {\em byzantine faults}. A faulty robot can fail to report
the target even after reaching it, or it can make malicious
claims about having found the target when in fact it has
not. Given the presence of such faulty robots, the search for the
target can only be concluded when the non-faulty robots have
sufficient verification that the target has been found.
We aim to design algorithms that minimize the value of
$S_d(n,f)$, the time to find a target at a distance $d$ from the
origin by $n$ robots among which $f$ are faulty.  We give several different
algorithms whose running time depends on the ratio
$f/n$, the density of faulty robots, and also prove lower bounds. Our algorithms are
optimal for some densities of faulty robots.

\vspace{0.5cm}
\end{abstract}

\newpage

\section{Introduction}

Searching on a line (also known as a {\em single-lane cow-path}
or a {\em linear search}) problem is concerned with a robot
looking for a target placed at an unknown location on an infinite
line; the robot moves with uniform (constant) speed and can
change direction (without any loss in time) along this line. The
ultimate goal is to find the target in optimal
time~\cite{baezayates1993searching}. Searching is central to many
areas of computer science including data structures,
computational geometry, and artificial intelligence. A version of
the problem
was first posed in 1963 by Bellman~\cite{bellman1963optimal} and
independently considered in 1964 by Beck~\cite{beck1964linear},
where the target was placed according to a known probability distribution on the real line, the robot was moving with uniform speed, and the goal was to find the target in minimum expected time.  


In this paper, we consider the problem of {\em parallel,
  co-operative} search on the infinite line by $n$ mobile robots
at most $f$ of which are faulty. The target is placed at a
distance unknown to the robots. The robots start at the same time
and location and can communicate instantaneously in wireless mode 
at any distance on the real line. While searching, the robots may co-operate by exchanging (broadcasting) messages; however, the search may be impeded by some of the robots (at most $f$) which may exhibit byzantine faults. The ultimate goal is to minimize the time it takes all non-faulty robots to be certain that the correct location of the target has been found.  

\subsection{Motion and communication model}

To begin, we describe the robots' locomotive and communication models used in a search algorithm. 

\vspace*{1.5mm}
\noindent
{\em Robots and their trajectories.}
Robots are assumed to start at a common location, considered to
be the origin of the line. They can move at maximum unit speed
either along the positive direction (described as moving
\textit{right}) or along the negative direction (described as
moving \textit{left}); any robot can change direction arbitrarily
often (by \textit{turning}) without any loss in time. An
algorithm for parallel search specifies a \textit{trajectory}
unique to each robot that is given by its turning points, and the
speed(s) to follow between turning points. Since each robot has a
distinct identity, it may also follow a distinct trajectory.
Robots are assumed to have full knowledge of all trajectories, and
moreover can communicate instantaneously with each other in
wireless mode at any distance.  Since robots know all the trajectories, the only kind of message broadcast by a robot $R$ is whether or not it has found the target at some location; if $R$ stays silent while visiting some location,
the implicit assumption made by the other robots is that $R$ did not detect the target there. Thus $R$ follows its predefined trajectory until either it finds the target, in which case it announces that it has found the target, or it hears some other robot $R'$ announce that it has found the target, at which point $R$ may change its trajectory to participate in a verification protocol in regard to the announcement. 
 
\vspace*{1.5mm}
\noindent
{\em Messages and communication.}
All $n$ robots know that $f$ of the robots are faulty but they
cannot differentiate in advance which among them are faulty;
instead they must distinguish faulty from non-faulty ones based
on conflict resolution and verification of messages received
throughout the communication exchanges taking place during the
execution of the search protocol. To this end, robots are
equipped with pairwise distinct identities which they cannot
alter at any time (in that respect our model is similar to the
weakly Byzantine agent in \cite{dieudonne2014gathering}). In
addition to the \textit{correct} identity, the current location
of a robot is automatically included in any broadcast message
sent by the robot. Consequently, a faulty robot that does not
follow its assigned trajectory can be immediately detected as
faulty by the other robots, if it chooses to broadcast at some
stage. In all other ways, faulty robots are indistinguishable
from non-faulty (\textit{reliable}) robots, except that the
former can make \textit{deliberate} positive and negative
detection \textit{errors} as follows.  A non-faulty or \textit{reliable}
robot never lies when it has to confirm or deny the existence of
the target at some location. Contrast this with {\em a faulty
  robot that may stay silent even when it detects or visits the
  target, or may claim that it has found the target when, in
  fact, it has not found it.} Thus, a reliable robot
\textit{cannot necessarily trust} an announcement that the target
has been found, nor can it be certain that a location - visited
silently by another robot - does not contain the target. In other
words, the search for a target can terminate  
only after at least one robot that is \textit{provably reliable} has visited the target and announced that it has been found. This requirement is critical to all our algorithms: if at some time, multiple robots make conflicting announcements at a location then the resulting (conflicting) \textit{votes} can only be resolved if something is known about the number of reliable robots that participated in the vote. For instance, if three robots vote and it is known that two of them are reliable, then the majority vote would be the truth. 



\subsection{Preliminaries and notation}

Consider a parallel search algorithm for a target located at distance $d$ from the origin. 
First we define the search time of the algorithm and its corresponding competitive ratio.

\begin{definition}[Search Time]
Let $S_d(n,f)$ denote the time it takes for a search algorithm using a collection of $n$ robots at most $f$ of which are faulty, to find in parallel the location of a target placed at a distance $d$ (unknown to the robots) from the starting position (the origin) of the robots on the line. 
\end{definition}

\begin{definition}[Competitive Ratio]
The corresponding \textit{competitive ratio} is defined as $S_d(n,f)/d$, which is the ratio of the algorithm's search time and the lower bound $d$ on the time taken by any algorithm for the problem. 
\end{definition}

For larger values of $n$ and $f$, it will be more convenient to express our results in terms
of the {\em density}, $\beta = \frac{f}{n}$, of faulty robots. This leads to the following definition.

\begin{definition}[Asymptotic Competitive Search Ratio]
Extend the definition of $S_d(n,f)$ above to non-integer values
of $n$ by replacing $n$ with $\lceil n \rceil$ while the parameter $f$
remains integral.  Let $\beta = \frac{f}{n}$. Then 
\begin{equation}
\label{asymp:def}
\hat{S}(\beta) = \min 
~\left\{ \alpha ~|~  \exists \mbox{ constant
} c_{\beta} 
\mbox{ such that } \forall f > 0, ~S_d\left(f/\beta + c_{\beta}, f \right)
\leq \alpha d 
\right\}
\end{equation}
denotes the asymptotic competitive search ratio of any algorithm
with search time $S_d(n,f)$.
\label{defn-ACSR}
\end{definition}

Note that if $n\geq 4f+2$, then in any partition of the robots into
two groups each of size at least $2f+1$, we will always
have at least $f+1$ reliable (non-faulty) robots per group. Therefore, an
algorithm that sends the corresponding robots in the two groups in opposite directions is
guaranteed to find the target in time $d$, because when the target
is visited by one of the groups, a straightforward majority vote
in the group confirms its presence reliably. Hence, $S_d(4f+2, f)
= d$, which is optimal. On the other hand, if $n \leq 2f$, there
is no algorithm to complete the search: the $f$ faulty robots 
may always completely disagree with the reliable ones, making
it impossible to be certain of the location of the target.  Therefore, in the sequel, we examine the interesting case where
$2f+1 \leq n \leq 4f+1$.

\subsection{Our results}

In Section~\ref{sec:ub}, we are concerned mostly with upper bounds. Subsection~\ref{sec:design} establishes the guiding principles for the design of algorithms. 

We begin our study of upper bounds in Subsection~\ref{sec:small-n} by establishing bounds for $S_d(n,f)$ for specific small values of $n$ and $f$. These results are summarized in Table \ref{t1}.  For a
comparison, we include in Table \ref{t1} known results on the
search time for algorithms on the line with faulty
robots that exhibit only \textit{crash} faults \cite{PODC16}, i.e.,
when the faulty robots never send any messages.  


\begin{table}[htb]
\caption{Upper and lower bounds on the search time
  $S_d(n,f)$ for a given number  $n\leq 6$ of robots and faults
  $f=1,2$. Byz UB and Byz LB denote the known upper  and lower
  bound for byzantine faults while Crash UB and Crash LB denote
  the known  upper and lower bound for crash faults.} 
\begin{center}
\begin{tabular}{| c ||| c | c  ||| c | c |}
\hline
$n,f$ & Byz. UB & Byz. LB & Crash-UB & Crash-LB \\
\hline
$3,1$ & $9d$ & $3.93d$ & $5.24d$ & $3.76d$\\
\hline
$4,1$ & $3d$ & $3d$ & $d$ & $d$ \\
\hline
$5,1$ & $2d$ &  $2d$ & $d$ & $d$ \\
\hline
$6, 1$ & $d$ &  $d$ & $d$ & $d$ \\
\hline
$5,2$ & $9d$ & $3.57d$ & $4.43d$ & $3.57d$ \\
\hline
$6,2$ & $4d$ & $3d$ &$d$ &$d$\\
\hline
\end{tabular}
\end{center}
\label{t1}
\end{table}

For larger values of $n$ and $f$ we express our results in terms of the density $\beta = \frac fn$ and show how to extend our algorithms from small values of $n$ and $f$ to this setting. Table \ref{t2} summarizes our results from Subsection~\ref{sec:large-n}. 
\begin{table}[htb]
\caption{Upper and lower bounds on the asymptotic competitive search ratio $\hat{S}(\beta)$ for 
 various  ranges of the density $\beta$. Note that for $\beta > \frac 12$ the search problem is impossible to solve.}
\begin{center}
 { \tabcolsep=1.2mm
\begin{tabular}{|c| c | c |c| c |c| c | c | c|c | c  |}
  \hline
  &&&&&&&&&&\\
$\beta$ & $\leq \frac14$&$(\frac{1}{4}, \frac{3}{10}]$ &$(\frac{3}{10},\frac{1}{3}]$&$(\frac{1}{3},\frac{5}{14}]$&
$(\frac{5}{14}, \frac{13}{34}]$&$(\frac{13}{34},\frac{19}{46}]$&
$(\frac{19}{46}, \frac{47}{110}]$&$(\frac{47}{110}, \frac{65}{146}]$& $(\frac{65}{146}, 
\frac{157}{396}]$&$(\frac{157}{396}, \frac{1}{2}]$ \\
& & & & & && & & &\\ 
\hline
UB& $1$&$2$&$3$&$3$&$4$&$5$&$6$&$7$&$8$&$9$\\
\hline
LB&$1$&$2$&$2$&$3$&$3$&$3$&$3$&$3$&$3$&$3$\\
\hline
\end{tabular}
}
\end{center}
\label{t2}
\end{table}%

Subsection~\ref{intrigue:sub} concludes Section~\ref{sec:ub} with several intriguing algorithms in that for densities $\frac f n$ between $\frac{3}{10}$ and $\frac 1 3$ the resulting search time is between $2d$ and $3d$.
In Section~\ref{sec:lb}, we derive two lower bounds on the search time.

\subsection{Related work}

A search problem is usually seen as localization of a hidden target using  searchers capable to move in the environment. It is an optimization question, usually attempting to minimize the time needed to complete the search. The question has been studied in numerous variations involving static or moving targets, one or many searchers, known or unknown environment,  synchronous or asynchronous settings, different speed agents and many others (cf. \cite{FT08}).

In several studies, when the environment is not known in advance, search implies exploration, often involving mapping and localizing searchers within the environment \cite{AH00,AKS02,DKP91,HIKK01,K94,PY}. However, even for the case of a known, simple environment like a line, there were several interesting studies attempting to optimize the search time. They started with the independent works of Bellman~\cite{bellman1963optimal} and Beck~\cite{beck1964linear}, in which the authors attempted to minimize the competitive ratio in a stochastic setting. More exactly, they proved that time $9d$ is needed to guarantee finding the target situated at a (\textit{a priori} unknown) distance $d$ from the origin. Several other works on linear search followed (e.g. see \cite{alpern2002theory,beck1964linear,beck1965more,beck1984linear,beck1970yet,beck1973return,bellman1963optimal}). More recently the search by a single searcher was studied for different models, e.g., when the turn cost was considered \cite{demaine2006online}, when the bounds on the distance to the target are known in advance \cite{Bose13}, and when the target was moving or for more general linear cost functions \cite{Bose16}.

Most recently variants of linear search were studied for collections of collaborating searchers (robots). \cite{Groupsearch} considered linear group search, when the process is completed when the target is reached by the last robot visiting it. The robots collaborate attempting to minimize the group search time. However, \cite{Groupsearch} shows that having many robots does not help and the optimal search time is still bounded from below by $9d$. Group search using a pair of robots having distinct maximal speeds was studied in \cite{SIROCCO16}, in which techniques producing optimal search time were designed.

Fault tolerance was studied in distributed computing in various settings in the past (e.g., see \cite{hromkovic,lamport,lynch}). However, the subject of unreliability was mainly for static components of the environment (e.g. network nodes or links), which was sometimes modelled by dynamically evolving environments (cf. \cite{casteigts2011,kuhn2010}). The malfunctions arising to mobile robots were investigated for various problems of gathering or pattern forming \cite{agmon2006fault,cohen-peleg-convergence,dieudonne2014gathering,souissi2006gathering} or patrolling \cite{ISAAC15}. Recently \cite{PODC16} investigated crash faults of robots performing linear search, where the time of finding the target by the first reliable robot was optimized. However, dealing with Byzantine agents is in general more tricky, requiring to identify and to refute the most malicious adversarial behavior (e.g., see \cite{dieudonne2014gathering}).

\section{Upper Bounds}
\label{sec:ub}

As already observed, if $n \geq 4f+2$, linear search can be
performed optimally in time $d$, and no algorithm exists if $n
\leq 2f$. Therefore, we consider below the case when $2f+1 \leq n
\leq 4f+1$. Clearly, the robots can always stay together as a
group, and perform the doubling zig-zag strategy that is optimal
for a single robot and that has competitive ratio 9 \cite{baezayates1993searching,beck1964linear}. Since the
reliable robots (at least $f+1$) are always in a majority, we are
guaranteed to find the target. This yields the following upper
bound: 
\begin{theorem}  
\label{th:ub9d}
$S_d (n,f) \leq 9d$.
\end{theorem}

In the remainder of this section, we provide upper bounds that,
in general, are better than those suggested by Theorem
~\ref{th:ub9d} for the search problem. We do so by identifying
and using some {\em guiding} principles to  design search
algorithms in the presence of faulty robots. 


\subsection{Principles for the design of algorithms}
\label{sec:design}
The general framework of our algorithms involves five basic principles, namely
{\em Partition into Groups, Symmetry of Algorithms, Resolution of Conflicts, Simultaneous Announcements}, and {\em Computations by the Robots}, which we describe below in  detail.

\vspace*{1.5mm}
\noindent
{\em Partition into Groups.}
Depending on the ratio of faulty robots, we partition the robots into a certain number of groups. Two of the groups lead the exploration in opposite direction from the origin of the line.
Further, each of these two groups will have at least $f+1$ robots so that
at least some of the robots would announce the target when it is reached.

\vspace*{1.5mm}
\noindent
{\em Symmetry of Algorithms.}
The algorithms are symmetric as far as left and right part of the line is concerned. We therefore typically discuss the behavior of the algorithm with respect to one side of the line only.  

\vspace*{1.5mm}
\noindent
{\em Resolution of Conflicts.}
If at any time there is an announcement of a target, the robots in the search groups stop until the claim is resolved.  In the meantime, robots from some other group(s) move to resolve the claim. Once the claim is resolved, either the target is found and the robots stop,  or a certain number of faulty robots is identified.
From this time onward, the algorithm disregards any message from these faulty robots, effectively  reducing the number of faulty robots to contend with, and the groups  continue the search. 
Thus, each such announcement exposes more of the faulty robots,
until eventually, we can be certain of a majority of robots in each search group being reliable, in which case
the remaining search can be easily finished.

\vspace*{1.5mm}
\noindent
{\em Simultaneous Announcements.}
When two announcements are being made at the same time, as usual with wireless transmissions, the algorithm deals only with one of them at a time, chosen arbitrarily. After the resolution of the first announcement is done and the search is possibly restarted, the robots redo their observation, and then the announcement is repeated if needed, 
thus taking into consideration the situation after the resolution of the first announcement.  We show it does not influence the search time.       

\vspace*{1.5mm}
\noindent
{\em Computations by the Robots}
We assume that the time spent on calculations is negligible in comparison with the time spent in moving. Thus, we count only the time needed in movements of the robots until the target is found.    

As indicated above, throughout the execution of the algorithms, conflicts will be resolved by voting. More precisely, we define $V(x, t)$ to be the {\em vote} of the robots about position $x$ at time $t$. If $y$ robots have claimed that the target is at $x$ at or before time $t$, while $z$ have claimed (by visiting and keeping silent) that it is not at $x$, then we say $V(x, t) = (y, z)$. 

\begin{definition}[Conflict]
We say there is a {\em conflict} at position $x$ at time $t$ if $V(x, t)= (y, z)$, with $0 < y,z\leq f$.
\end{definition}
The following two simple observations are used extensively in the proofs in this section. 

\begin{lemma}
\label{vote}
Let $V(x, t)= (y, z)$, and let $f$ be the number of faulty robots before time $t$. Then
\begin{enumerate}
\item If $y > f$ then the target is at position $x$ and the search is concluded.
\item If $z > f$ then the target is not at position $x$ and $y$ new faulty robots have been identified at time $t$. 
\end{enumerate}
\end{lemma}

\begin{lemma}
\label{finish}
Suppose at time $t$, there are $f'$ faulty robots remaining, and there are at least $2f'+1$ robots at positions $\geq x$ and at least $2f'+1$ robots at positions $\leq -x$. Then any target that is distance $d$ from the origin can be found in time $t  + (d-x)$.  
\end{lemma}

To build intuition, we start with giving algorithms
with at most 2 faulty robots, and later show how to use
these techniques to give algorithms with asymptotic ratios for general values of $n$ and
$f$. 



\subsection{Algorithms for $n \leq 6$ }
\label{sec:small-n}

Since $2f+1 \leq n < 4f+2$, there are only two kinds of possible combinations of
values with $n \leq 6$:
either $f=1$ and $3 \leq n \leq 5$, or
$f=2$ and $5 \leq n \leq 6$. 

\begin{proposition}
\label{pr:3ds}
$S_d (4,1) \leq 3d$ 
\end{proposition}
\begin{proof} (Proposition~\ref{pr:3ds})
Partition all robots into two search groups, $L$, and $R$, with two robots in each group. Each robot in $R$ ($L$) moves right (left resp.) at speed 1 until it finds the target or hears an announcement that the target has been found. Suppose now that there is an announcement at time $x$ from position $x>0$. If $V(x, x) = (2, 0)$, by Lemma~\ref{vote}, the target has been found at $x$ and the algorithm terminates. 
Suppose that $V(x, x) = (1, 1)$.  Then  one of the robots in $L$,
say $A$, travels to $x$ to resolve the conflict, taking
additional time $2x$, while all other robots remain stationary. At time $3x$, the robot $A$ reaches $x$. If 
$V(x, 3x) = (2, 1)$,  by Lemma~\ref{vote}, the target has been
found, and the algorithm terminates. If instead $V(x, 3x) = (1,
2)$, then by Lemma~\ref{vote}, the faulty robot is identified,
and all other robots can be  inferred to be reliable. Now the
search  continues with the groups moving in opposite directions 
with only the reliable robots being considered, until the target is found. 
Notice that an announcement at $-x$,  simultaneous with that at $x$, would be resolved at time $3x$ with reliable robots. Therefore, if the target is at $d$ or $-d$, the time taken to find it is $\leq 3x + d- x= 2x + d \leq 3d$ since $d > x$. Thus in all cases, $S_d(4, 1) \leq 3d$. 
 \end{proof}

If the number of robots increases to $n=5$ (while $f$ still
equals 1), then it is possible to send two groups of size $2$ in opposite
directions as in the algorithm above, but keep one \textit{spare} robot
at the origin for conflict resolution. This improves the search
time to at most $2d$ since the spare is always at a distance $d$
from a conflicting vote, and moreover, the spare is definitely
reliable since the faulty robot is part of the conflicting vote.. 
\begin{proposition}
\label{pr:3dsa}
$S_d (5,1) \leq 2d$ 
\end{proposition}

\begin{proof} (Proposition~\ref{pr:3dsa})
 Partition the robots into three groups, two of them, $L$ and $R$
 being the search groups. There are two robots in $L$ which moves
 left at speed 1,  one robot (designated as $A$ for convenience) remains stationary  at the origin, and one search group of two robots $R$  moves right at speed 1. 
Suppose that there is an announcement at time $x$ from position $x > 0$. 
If $V(x, x) = (2, 0)$, then by Lemma~\ref{vote},  the search
terminates in time $d$. If instead $V(x, x) = (1, 1)$,  then the
robot $A$ joins $R$ at position $x$, taking additional time $x$ to resolve the conflict, while all other robots remain stationary.  If $V(x, 2x) = (2, 1)$,  by Lemma \ref{vote},  the target has been found, and the algorithm terminates. If instead $V(x, 2x) = (1, 2)$, then by Lemma \ref{vote}, the faulty robot  is identified, and all remaining robots can be  inferred to be reliable. Now each group $L$ and $R$  with reliable robots continues moving in opposite directions until the target is found.  Notice that an announcement at $-x$,  simultaneous with that at $x$, would be resolved at time $2x$ with reliable robots. 
If the target is at $d$ or $-d$, the time taken to find it is $\leq 2x + d- x= x + d  \leq 2d$ since $d > x$.Thus in all cases, $S_d(5, 1) \leq 2d$.
\end{proof}

Note that the cases, $(n, f)=(5,2)$ or $(n,f)=(3,1)$, satisfy
$n=2f+1$, the bare minimum of robots necessary to guarantee
termination. For these cases, it seems very difficult to improve
upon the upper bound on $S_d(n,f) \leq 9d$ from 
Theorem ~\ref{th:ub9d}. In fact, we conjecture that this best
possible for the pairs $(5,2)$ and $(3,1)$ stated above. 

By ensuring an appropriate \textit{redistribution} of robots past
the announcement of a conflict, we can show the following result:
\begin{proposition}
  \label{pr:ub4.24}
$S_d (6,2) \leq 4 d$.
\end{proposition}

\begin{proof} (Proposition~\ref{pr:ub4.24})
We partition the robots into two groups of size three and have the two search groups move at speed 1 from the origin in opposite directions.  Assume that when the groups reach $-x$ and  $x$ for some $0 < x \leq d$  an announcement of a target is sent from $-x$. We only need to consider the cases when the  vote is $V(-x,x)=(2,1)$ or  $(1,2)$. 

{\it Case 1: } $V(-x,x)=(2,1).$ In this case we send two robots from $x$ to $-x$ and redo the vote. There
are three possibilities: $V(-x,3x) = (4,1); (3,2); (2,3).$ In the first two cases the target has been found
and we are done in $3d$. In the last case, the two ``yes" voters are the two faulty robots and may be
eliminated. Since we now have only good robots (three at $-x$ and one at $x$) we finish in $d-x$ for 
a total of $d + 2x \leq 3d$. 

{\it Case 2:} $V(-x,x)=(1,2).$ Again we send two robots from $x$ to $-x$ and redo the vote. Simultaneously, 
we send one of the ``no" voting robots to the origin and one ``yes" and one ``no" voter to $x$. There
are three possibilities for $V(-x,3x): (3,2), (2,3), (1,4)$. If the vote is $(3,2)$ we have found the
target and are done. If the vote is $(2,3)$, the two ``yes" votes are the faulty robots and may
be eliminated. There is at least one good robot at $x$ and at $-x$ and so the search may be completed
in $d-x$ for at total of less than or equal to $3d$. 

In the final case, we have two robots who voted ``no" at $-x$ one of which may be faulty and
we have one robot that voted ``no" plus one robot that hasn't voted at $x$ (plus the ``yes" voter
who may be eliminated) one of which may be faulty. I.e.,  at each of $x$ and $-x$ we have two robots at most one
of which may be faulty plus one robot at the origin which may be faulty. We complete the 
search as in the case of five robots one of which is faulty: I.e. at the next conflict, say
at $x' \leq d$, the center
robot resolves the conflict and reveals the last faulty robot at the cost of an additional $x'$
time. The search now completes in time $d - x'$ for a total of $3x+(x'-x) + x' + d - x' \leq 4d.$
This completes the proof.
\end{proof}


\subsection{Algorithms for large $n$}
\label{sec:large-n}

We now consider the case of large $n$, with different values of
the density,  $\beta = f/n$, of 
faulty robots.   We start with generalizing the results from the previous subsection, then build  recursive techniques that allow us to deal with larger 
densities of faulty robots, while paying a price in terms of the search time.

\begin{theorem}
\label{th:ub2d}
$S_d \left(\frac{10f+4}{3}, f \right) \leq 2d$, provided that $f \equiv 2 \mod{3}$.
\end{theorem}

\begin{proof} (Theorem~\ref{th:ub2d})
Partition the robots into three groups $L$, $R$, and $C$ of sizes $\frac{4f+1}{3}$, $\frac{4f+1}{3}$
and $\frac{2f+2}{3}$ respectively. 
 Robots in $L$ and $R$, are the search groups that  move left and  right at speed 1 respectively, and robots in $C$ stay stationary 
until an announcement is made.  Suppose there is an announcement at  time $x$ at position $x > 0$, and let $V(x, x) = (y, z)$ with $y>0$.
There are three possible cases. 
If $y > f$, then by Lemma~\ref{vote}, the search is concluded in time $x = d$. 
If $z> f$, the target is not at $x$, and $y$ faulty robots have been identified. In this
case, we disregard these robots from now on and the remaining robots in $L$ and $R$ continue moving in their original directions until
a genuine conflict occurs (i.e., both $y$ and $z$ are less than the number of faulty robots remaining) or
the target is found. 
  
It remains to consider the case when $\max \{ y, z \} \leq f$. Observe that $\min \{ y, z \} \geq \frac{f+1}{3}$.  In this case, the algorithm moves the robots in $C$ in time $x$ to location $x$. Consider now $V(x, 2x) = (y', z')$. Since at time $2x$ there are  $2f+1$ robots at $x$, 
either the search is concluded or at least $y' \geq y \geq \frac{f+1}{3}$ robots faulty robots have been identified and may be disregarded.
 
Consider the remaining robots; there are at most $f' = f - y'$ faulty robots among them. There are $2f + 1 - y'$ robots (after the faulty $y'$ have been identified) remaining in the group $R$, and 
  $\frac{4f+1}{3}$ robots in group $L$.
  It is easy to verify that with  $\frac{f+1}{3} \leq y' \leq f$, we have at least $2f'+1$ robots in both $L$ and $R$. Therefore, the robots in $L$ and $R$ can continue exploration of the line  in their original directions and the time required to reach the target is at most $2x + (d-x) \leq 2d$. Observe that all the quantities
  mentioned are integral in the case that $f \equiv 2 \mod{3}$. This completes the proof.
\end{proof}

Using the fact that $S_d(n+k,f) \leq S_d(n,f)$ for any $k \geq 0$ and that 
$\hat{S}(\beta) \leq \hat{S}(\beta')$ if $\beta \leq \beta'$ we can easily derive the
following corollary:

\begin{corollary}
  If $\beta \leq \frac{3}{10}$ then $\hat{S} (\beta) \leq 2$.
\end{corollary}

\begin{theorem}
\label{th:ub3d}
   $S_d\left( \frac{14f +4}{5}, f \right) \leq 3d$ provided $f \equiv 4  \mod{5}$.
\end{theorem}
\begin{proof} (Theorem~\ref{th:ub3d})
Partition the robots into two search groups $L$ and $R$ each containing 
$\frac{7f +2}{5}$ robots. The robots in $L$ move left and those in $R$ move right at speed 1. Without loss of generality, assume there is an announcement at $x$ at time $x$. Let $V(x, x) = (y, z)$. Then if 
$\max \{y, z \} > f$, the announcement is resolved using Lemma \ref{vote}. Suppose instead that $\max \{y, z \} \leq f$. Then 
$\min \{y, z \} \geq \frac{2f+2}{5}$ and at least $\frac{2f+2}{5}$ robots at $x$ are faulty. In this case, 
$\frac{3f+3}{5}$ robots from $L$ move from $-x$ to $x$, and at the same time 
$\frac{2f+2}{5}$ robots that voted yes and $\frac{2f+2}{5}$ that voted no are sent from $x$ to $-x$.   
At time $3x$, in total $2f+1$ robots have voted at $x$, and by Lemma~\ref{vote}, either the target is identified, or at least 
$\frac{2f+2}{5}$ faulty robots are identified at $-x$ and may be disregarded from now on. There are at most $\frac{3f-2}{5}$ faulty robots unidentified. After the exchange of robots and elimination of the faulty robots in
the worst case there are $\frac{6f+1}{5}$ robots in $L$ and in $R$, i.e., a majority of reliable robots in both search groups. 
Therefore by Lemma~\ref{finish}, search for a target at distance $d$ can be finished in time $3x+ d-x \leq 3d$ as claimed. Note that all quantities are integral if $f \equiv 4 \mod{5}$.
\end{proof}

As above, the following corollary is immediate:

\begin{corollary}
  If $\beta \leq \frac{5}{14}$ then $\hat{S} (\beta) \leq 3$. 
\end{corollary}

As illustrated in the proofs of Theorems \ref{th:ub2d} and \ref{th:ub3d}, when an announcement of a target is made, 
either the target can be confirmed, or the number of unidentified faulty robots can be reduced by an exchange of robots between the two search groups.
 For higher densities of faulty robots this technique can be repeated, for which we pay by an increase in the search time. This is the motivation for the    
recurrence formulas below that are used to obtain search algorithms for higher densities of robots. 
\begin{definition}
Let $T_x(l,s,r, f)$ be the minimum search time required by the robots to find the  target given that initially, 
$l$ robots are located at $-x$,  $s$ robots are at the origin $0$, $r$ robots are at $+x$, and $f$ robots are faulty. 
\end{definition}

Since, as in the algorithms described so far, one way to solve our search problem is to send two equal-sized groups of robots to positions $x$ and $-x$, we get the following upper bound. 

\begin{lemma}
\label{lemmaTx}
For all $d \geq x >0$, we have $S_d(n, f)  \leq x + T_x (n/2, 0, n/2, f)$. Furthermore, if $n/2 \geq 2f+1$ then
$T_x(n/2, 0, n/2, f) = d-x$.
\end{lemma}

If there is an announcement at $x$, we can identify some of the
faulty robots, and by paying a price in terms of additional time,
we can reduce it to a new problem with a smaller number of faulty robots. This can be encapsulated
in the following lemma:

\begin{lemma}
\label{recurrence2}
Let $k>0$ be even. Suppose there is an announcement at distance $x$ from the origin. Then for all $a \geq x$,
$T_x(f+k,0,  f+k, f) \leq 2x+ T_a(f+k/2, 0,  f+ k/2, f-k).$
\end{lemma}
\begin{proof} (Lemma~\ref{recurrence2})
Assume there are $f+k$ robots each at $x$  and $-x$, with at most $f$ faulty robots in all, and  that a conflict occurs at $x>0$ at some time $t$.  
Let $V(x, x) = (y, z)$. Then $k \leq  \min \{y, z \} \leq \max \{y, z \}  \leq f$.  Now the robots move as follows:
\begin{enumerate}
\item All $f+k$ robots at position $x$ move to $-x$.
\item $f+ k/2$ of the robots at $-x$ move to $x$.
\end{enumerate}
Note that these movements take time $2x$, and there are now
$f+k/2$ robots at $x$ and $f + k/2 + k$ robots at $-x$. Since
$2f+ 3k/2$ robots have now visited $x$, the vote $V(x, 3x)$ is
enough to resolve the conflict, and there remain at most $f-k$
faulty robots among the total $2f+k $ robots. This proves the
lemma.
\end{proof}

Lemma \ref{recurrence2} along with Theorem \ref{th:ub3d} can be used to obtain slower algorithms
for higher densities. We have

\begin{theorem}
\label{th:ub5d7d}~
\begin{enumerate}
\item $\hat{S}(\beta) \leq 5$ for $\beta \leq 19/46$.
\item $\hat{S}(\beta) \leq 7$ for $\beta \leq 65/146$.
\end{enumerate}
\end{theorem}

\begin{proof} (Theorem~\ref{th:ub5d7d})
An examination of the proof of Theorem \ref{th:ub3d} shows that it may reformulated
using the following observations:

\noindent
(a) For $x\leq d$, 
$$S_d\left(\frac{14f + 4}{5},f\right) \leq x + T_x\left(\frac{7f+2}{5}, 0, \frac{7f+2}{5}, f \right)$$
(b) Suppose that an announcement occurs at $x \leq d$. Then for 
all $a\geq x$, 
$$T_x\left(\frac{7f+2}{5},0,\frac{7f+2}{5},f\right) 
\leq 2x + T_a\left(\frac{6f+1}{5}, 0, \frac{6f+1}{2}, \frac{3f-2}{5} \right)$$
(c) For $x \leq d$, 
$$T_x\left(\frac{6f+1}{5}, 0, \frac{6f+1}{5}, \frac{3f-2}{5}\right) = d-x$$

Points (a) and (c) above follow from Lemma \ref{lemmaTx} and
point (b) follows from Lemma \ref{recurrence2} taking $k=\frac{2f+2}{5}$.  Since an announcement
must occur either before or at time $d$, together the above imply 
$S_d\left(\frac{14f+4}{5} , f\right) \leq 2x + d$ for some $x\leq d$.  To extend this to higher densities
we apply Lemma \ref{recurrence2} multiple times. 

Consider an $f$ that is a multiple of 19 and let $n = \frac{46f}{19} +12$. By Lemma
\ref{recurrence2} we may conclude that if there is an announcement at some $x \leq d$, then
for all $a \geq x$, 
$$T_x \left(\frac{23f}{19} +6,0,\frac{23f}{19}+6,f \right) \leq 2x +
T_a\left(\frac{21f}{19}+3, 0, \frac{21f}{19}+3, \frac{15f}{19} -6 \right)$$ 
taking $k = \frac{4f}{19} +6$. Observe that $\frac{15f}{19} -6 = 4 \mod 5$, 
$\frac{7 (\frac{15f}{19} -6) +2 }{5} = \frac{21f}{19}-8$ and that 
$$T_x\left(\frac{21f}{19}+3, 0, \frac{21f}{19}+3, \frac{15f}{19} -6\right)) \leq
T_x\left(\frac{21f}{19}-8, 0, \frac{21f}{19}-8, \frac{15f}{19} -6 \right)$$
holds. Thus after a single announcement, 
we have reduced our problem to one that may solved using the approach from Theorem 
\ref{th:ub3d}. 

From this we can conclude that for $f$ a multiple of 19 and $n =
\frac{46f}{19} +12$,  if the announcements  
occur at $x_0 \leq x_1 \leq d$ we have 
$$S_d\left(\frac{46f}{19} + 12, f \right)) \leq  d + 2x_0 + 2x_1+1 \leq 5d,$$
from which it follows that $\hat{S}(\beta) \leq 5$ for $\beta \leq 46/19$. Applying the lemma
one more time, we can show  $\hat{S}(\beta) \leq 7$ for $\beta \leq  146/65.$
\end{proof}

Similar to Lemma \ref{recurrence2} the following lemma establishes a recurrence that can be
used to extend Theorem \ref{th:ub2d} to higher densities (at a cost of a higher competitive ratio).

\begin{lemma}
\label{recurrence1}
Suppose there is an announcement at distance $x$ from the origin. Then for all $a \geq x$ and $k \geq f/4$:
$T_x(f+k, 0, f+k, f) \leq 2x+ T_a \left(\frac{4(f-k)}{3}, \frac{2(f-k)}{3}, 3k, f-k\right).$
\end{lemma}

\begin{proof} (Lemma~\ref{recurrence1})
Assume there are $f+k$ robots each at $x$  and $-x$, with at most $f$ faulty robots in all, and  that a conflict occurs at $x$ at some time $t$. Let $V(x, t) = (y, z)$. Then $k \leq  min \{y, z \} \leq max \{y, z \}  \leq f$.  Now the robots move as follows:
\begin{enumerate}
\item  $f-k$ robots move from $-x$ to $x$.  
\item $\frac 43 f - \frac{10}{3} k$ robots move from $x$ to $-x$.
\item $\frac 23 (f-k)$ move from $x$ to 0. 
\end{enumerate}
Note that these actions take $2x$ time, and it is clear that there remain 
$\frac{4(f-k)}{3}$, $\frac{2(f-k)}{3}$, $3k$ robots at $-x$, $0$, and $x$, respectively. At this time, we take another vote at position $x$, adding to it
the votes of robots that were at $x$ at time $t$ but were relocated. Observe that the total number of votes tallied is $2f$, and therefore the conflict can be resolved.  Furthermore, by Lemma~\ref{vote}, at least $k$ faulty robots have been identified, so the number of faulty robots that remain is $f-k$. 
\end{proof}

%

Using a similar argument to that used in Theorem \ref{th:ub5d7d} we can apply Lemma
\ref{recurrence1} and Theorem \ref{th:ub2d} to get:

\begin{theorem}
\label{th:ub4d6d8d}~
\begin{enumerate}
\item $\hat{S}(\beta) \leq 4$ for $\beta \leq 13/34$.
\item $\hat{S}(\beta) \leq 6$ for $\beta \leq 47/110$.
\item $\hat{S}(\beta) \leq 8$ for $\beta \leq 157/396$.
\end{enumerate}
\end{theorem}

\subsection{Algorithms for $\frac{3}{10} \leq \beta < \frac{1}{3}$ }
\label{intrigue:sub}


Finally we discuss a new class of algorithms for densities of $\frac f n$ between
$\frac{3}{10}$ and $\frac 1 3$ whose search time is between $2d$ and $3d$.

Informally, in any of these algorithms, the robots are partitioned into two search groups, that move in opposite directions at speed 1, and $i$ middle groups,
$i$ odd, $i\geq 3$, positioned at regular intervals between the search groups. 
These $i$ groups are used to solve any conflict reached by the search groups.
The positioning of the middle groups between the search groups is achieved by them moving at a fraction of the maximal speed.

When a vote arises that cannot be resolved using Lemma \ref{vote}, the middle groups are moved to the point of conflict in sequence at speed 1 until a resolution of the conflict is obtained. The middle groups not used in the resolution of a conflict on one side can be used to resolve a conflict on the other side. This approach allows a fine-grain resolution of a conflict by taking into account the result of the vote each time a group arrives to the conflict point.

\begin{lemma} 
\label{fraction:lm}
Let $i$ be an odd integer, $i\geq 3$. \\
$S_d(\frac{(3i+2)f}{i+1}+2,f)) \leq \left(3-\frac{2}{i+1} \right)d$, provided
$f \equiv 0 \mod(i+1).$
\end{lemma}

\begin{proof} (Lemma~\ref{fraction:lm})
In our algorithm we partition the robots into $i+2$ groups. Two of these groups, the search groups,  are of size
  $\frac{i+2}{i+1}f+1$ and they move at speed $1$ in opposite directions.
Each of the remaining $i$ groups, called middle groups,
are of size $\frac{1}{i+1}f$ and they move at speed less than $1$ so that when the search groups are located at $-x$ and $x$, 
they are located at points $-d + \frac{2d}{i+1}, -d + 2\frac{2d}{i+1},-d + 3\frac{2d}{i+1}, \ldots, d-2i\frac{2d}{i+1}, d-\frac{2d}{i+1}$ between $-x$ and $x$, 
i.e., $i$ points that divide the interval $(-x,x)$ into equal size segments.
Assume that at time $x$ there is a conflicting vote $V(x,x)=(y,z)$ at point $x$ on the line that cannot be resolved using Lemma \ref{vote}, i.e.,
$\max\{y,z\} \leq f$ and thus $\min\{y,z\}\geq \frac{f}{i+1}$. We start to shift
the middle groups to $x$ at full speed, and observe the vote after the arrival of
each middle group. If after the arrival of $j$ middle groups the conflict is not solved, then at that point  $V(x,x+2x/(i+1)j)=(y_j,z_j)$ with  $\min\{y_j,z_j\}\geq \frac{(j+1)f}{i+1}$ and thus the group of robots now located
at $x$ contains at least  $\frac{(j+1)f}{i+1}$ faulty robots and the search group on the left contains at most $\frac{i-j}{i+1}$ faulty robots. Thus we can conclude that after arrival of $(i-1)/2$ middle groups to $x$, the search group on the right contains the majority of reliable robots and the middle groups are not needed in solving conflicts on the left. Observe that prior to arrival of $(i-1)/2$ robots groups to $x$, the middle groups located to the left of $0$ contain sufficient number of robots to solve any conflict at $-x$. On the other hand, after the arrival of all $i$ middle groups to $x$ we have $2f+1$ robots there which resolves the conflict, and sufficient number of faulty robots is identified to continue the search with the majority of reliable robots on the right as well.
Thus in time $x+(2x-\frac{2x}{i+1})= 3x-\frac{2x}{i+1}$ the search is either finished or can be done in time at most $d-x+ 3x-\frac{2x}{i+1}\leq (3-\frac{2}{i+1})d$. 
\end{proof}

\begin{corollary}~
\begin{enumerate}
\item $\hat{S}(\beta) \leq 2.5$ for $\beta \leq 4/13$.
\item $\hat{S}(\beta) \leq 2.67$ for $\beta \leq 6/19$.
\item $\hat{S}(\beta) \leq 2.75$ for $\beta \leq 8/25$.
\item $\hat{S}(\beta) \leq 2.8$ for $\beta \leq 10/31$.
\end{enumerate}
\end{corollary}

%


\section{Lower Bounds}
\label{sec:lb}

It is  straightforward to see that to achieve search time $d$, $4f+2$ robots are necessary; with $4f+1$ or fewer robots, at time $d$, either $d$ or $-d$ can be visited by at most $2f$ robots. The adversary can make $f$ of these $2f$ robots faulty, and it is impossible to be certain about the answer. Formally we can prove the following result.

\begin{lemma}
\label{2dlb:lm}
  $S_d (5,1) \geq 2d$.
\end{lemma}
\begin{proof} (Lemma~\ref{2dlb:lm})
At time $d-\epsilon$ no one has visited $d$ or $-d$. Consider where the robots are at this time. It must be the case that one of the intervals $(-d,0)$ or $(0,d)$ contains at most $2$ robots. Without loss of generality  say it is $(0,d)$. Put the target at $d$. Sort the robots by distance to $d$ (ties broken arbitrarily) and  make the robot closest to $d$ faulty and silent. Then at least one robot from $(-d,0]$ must also reach $d$ so that two  non-faulty robots can identify the target at $d$. Thus, the search time is at least
  $d-\epsilon + d = 2d-\epsilon$. 
\end{proof}

The next theorem shows that the density $f/n=\frac{3}{10}$ in Theorem \ref{th:ub2d} is also a lower bound on this ratio if we want to maintain the search time to be at most $2d$. 
\begin{theorem}
  \label{th:lb2d}
  If $S_d(n,f) \leq 2d$ then $\frac{f}{n} \leq \frac{3}{10}$.
\end{theorem}
\begin{proof} (Theorem~\ref{th:lb2d})
Assume on the contrary that $\frac{n}{f} \leq \frac{10}{3} - \epsilon$ and that there is an algorithm for solving the search problem in time $2d$. Observe the intervals $[-d, 0), ~\{ 0\},~ (0, +d]$ at time $d$ and let us denote by $l, r$ the number of robots within $[-d, 0)$, $(0, +d]$, and by $s$ the number of robots at the origin $0$, respectively. By assumption $l+r+s = (10/3-\epsilon)f$. 
Observe that robots which are located at points different from $-d,0,d$ at time $d$ may not be helpful in reducing the $2d$ search time. Thus, without loss of generality we may assume that at time $d$ only the points $-d, 0, +d$ are occupied by robots. Without loss of generality assume that $r \leq l$.
We derive a contradiction by considering two cases.
\begin{enumerate}
\item Either $l$ or $r \geq \frac{4}{3} f$. In this case we have that $r+s=\frac{10}{3}f-\epsilon -l \leq \frac{10}{3}f-\epsilon-
 \frac{4}{3}f=2f-\epsilon$. 
Thus, $s+r$ robots  are not sufficient to resolve conflicts on the right 
possibly involving $f$ faulty robots within time $2d$.
\item Assume that there exists $\epsilon >0$ such that  both, $l,
  r \leq (\frac{4}{3} -\epsilon )f$. 
In particular,  consider $ r \leq (\frac{4}{3} -\epsilon )f$.
Consider time $d$ and suppose that up to $min\{r,\frac{1}{3} f\}$ of robots at $d$ claim to find the target. 
For the algorithm to attain time $2d$, 
robots must be send from the start position $0$ at time $d$ to position $d$ so as to verify the claim. Since among the robots sent to $+d$ from $0$ we could have all remaining faulty robots, the number of robots sent from $0$ must be at least $2f+1-r$ so that we a decision at time $2d$ can be made.  
However, if the target  is not at $+d$ then the adversary could make it so that only $\frac 1 3 f$ robots are faulty at $+d$ from among $2f+1$ robots. However, now we have at most 
$\frac{10}{ 3}f -\epsilon -2f -1=\frac 4 3 f -\epsilon -1$ robots at $0$ or to the left of $0$  and still $\frac 2 3 f$ faulty robots remain among them. Thus, any claim of target  at $-d'$ to the left of $-d$ cannot be verified in time $2d'$ 
by the available robots. 
\end{enumerate}
This proves the theorem.
\end{proof}

\begin{lemma}
\label{3dlblm:lm}
$S_d (3f+1, f) = 3d$.
\end{lemma}
\begin{proof} (Lemma~\ref{3dlblm:lm})
The upper bound $S_d (3f+1, f) \leq 3d$ has been proved in
Theorem~\ref{th:ub3d}.  To prove the lower bound $S_d (3f+1, f)
\geq 3d$ we argue as follows. Consider visits to the set of
symmetric positions $\{-d, +d\}$ by the robots. In
particular, consider the first time $t$ that at least $f+1$ 
robots complete visits to the \textit{second} of the positions in the
set. For instance, without loss of generality, assume that
position $-d$
is visited first by at least $f+1$
robots and later (or instantaneously) by at least 
$f+1$ robots. Clearly the time $t$ is at least $d$. The adversary
arranges for a conflict at position $+d$. Note that unless $t \geq 3d$, the
sets of robots visiting the two positions must be disjoint, and
hence, the conflict at position $+d$ involves at most $2f$ robots
participating in a vote, i.e. to resolve the conflict, at least
one of the robots that visited $-d$ must move to $+d$. It follows
that the total time required is at least $t+2d \geq 3d$.
\end{proof}

Note that Lemma~\ref{3dlblm:lm} implies a lower bound for densities $\beta$ in the range $1/3>\beta>3/10$. In case $n=3$, $f=1$, we can show the following lower bound on the search time. 
\begin{lemma}
\label{3.93lb:lm}
  $S_d(3,1)\geq  3.93d$.
\end{lemma}
\begin{proof} (Lemma~\ref{3.93lb:lm})
We start by considering three positive real numbers $x, y, \alpha$ such that
\begin{equation}
\label{max:eq}
\frac{\alpha -1}{2} \leq x < y \leq \frac{2}{\alpha -3}
\mbox{ and }
\frac{\alpha -1}{2} \leq \frac{y}{x} \leq \frac{2}{\alpha -3} .
\end{equation}
We will show that an $\alpha$ satisfying Inequalities~\eqref{max:eq} above is the competitive ratio of all search algorithms for three robots with one Byzantine fault. Moreover, using Mathematica it can be shown that the maximum value of $\alpha$ that satisfies \eqref{max:eq} is $3.93$.

Consider numbers $-y,\  -x,\  -1,\  0,\  1,\  x,\  y$
on the real line  
and the movement of the three robots with respect to these points. Assume on the contrary the competitive ratio is some value $\rho$ such that $\rho < \alpha$. Throughout the arguments below we are using Inequalities~\eqref{max:eq}.

Observe that two robots must visit the points $-1, 1$ before time $\alpha$, otherwise we get a contradiction to the competitive ratio because of Inequality~\eqref{max:eq}. Therefore there exists a robot, say $A$, that visited both of these points before time $\alpha$.  Same argument applies for points $-x, x$. There exist a robot that visits both points $-x, x$ before time $\alpha x$. Observe that this robot cannot be $A$. Indeed, otherwise it takes either time $2x+1$ to reach point $-1$ or time $2+3x$ to reach point $x$. Let $B$ be the robot that visits both points $-x, x$ before time $\alpha x$. Because of the time constraints in Inequalities~\eqref{max:eq} the robot $B$ must have either positive  trajectory (i.e., visiting $x$ before $-x$) or negative  trajectory (i.e., visiting $-x$ before $x$). However, it is easy to see that $B$ cannot have a positive trajectory because it would be too far to confirm an target  placed at $-1$. This proves the lemma
\end{proof}


\section{Discussion}

\label{sec:discussion}

In this paper, we considered a generalization of the well-known
cow-path problem by having the search done in parallel with a
group of $n$ robots, with up to $f$ of them being byzantine
faulty.  We presented optimal search algorithms for several
ranges of values for $\beta = f/n$, the fraction of faulty
robots, and gave non-trivial upper and lower bounds in many
cases. Several interesting problems in the setting remain open,
the most interesting one being to give tight upper and lower
bounds in the case $n=2f+1$.

\bibliographystyle{abbrv}
\bibliography{refs}

\end{document}